\documentclass[11pt]{article}
\usepackage{authblk}
\usepackage{amsfonts,amsmath,amsthm}
\usepackage[colorlinks=true, allcolors=blue]{hyperref}
\usepackage{mathabx}

\usepackage{fullpage}

\title{Lower Bounds on Error Exponents via a New Quantum Decoder}

\author[1,2]{Salman Beigi}
\author[2,3]{Marco Tomamichel}
\affil[1]{\it \footnotesize School of Mathematics, Institute for Research in Fundamental Sciences (IPM), Tehran, Iran}
\affil[2]{\it \footnotesize Centre for Quantum Technologies, National University of Singapore, Singapore}
\affil[3]{\it \footnotesize Department of Electrical and Computer Engineering, National University of Singapore, Singapore}

\date{}

\newcommand{\cX}{\mathcal{X}}

\newcommand{\cM}{\mathcal{M}}
\newcommand{\cC}{\mathcal{C}}
\newcommand{\cW}{\mathcal{W}}
\newcommand{\cN}{\mathcal{N}}
\newcommand{\cH}{\mathcal{H}}

\newcommand{\cP}{\mathcal{P}}
\newcommand{\cE}{\mathcal{E}}
\newcommand{\E}{\mathbb{E}}
\renewcommand{\d}{\textnormal{d}}
\newcommand{\proj}[1]{| #1 \rangle\langle #1 |}
\newcommand{\ket}[1]{|#1\rangle}
\newcommand{\bra}[1]{\langle#1|}
\newcommand{\ketbra}[2]{|#1\rangle\langle#2|}

\newcommand{\altfrac}[2]{#1:#2}

\newtheorem{theorem}{Theorem}
\newtheorem{lemma}[theorem]{Lemma}

\DeclareMathOperator{\tr}{tr}
\DeclareMathOperator{\spec}{spec}
\DeclareMathOperator{\poly}{poly}

\begin{document}

\maketitle

\begin{abstract}
    We introduce a new quantum decoder based on a variant of the pretty good measurement, but defined via an alternative matrix quotient. We then use this novel decoder to derive new lower bounds on the error exponent both in the one-shot and asymptotic regimes for the classical-quantum and the entanglement-assisted channel coding problems. Our bounds are expressed in terms of measured (for the one-shot bounds) and sandwiched (for the asymptotic bounds) channel R\'enyi mutual information of order between $1/2$ and $1$. The bounds are not comparable with some previously established bounds for general channels, yet they are tight (for rates close to capacity) when the channel is classical. Finally, we also use our new decoder to rederive Cheng's recent tight bound on the decoding error probability, which implies that most existing asymptotic results also hold for the new decoder.
\end{abstract}

\section{Introduction}

Coding over classical-quantum (cq) channels is generally very well-understood. Such channels are characterised by a mapping $\cW: x \mapsto \rho_x$ for an ensemble of quantum states $\{ \rho_x:\, x \}$. The capacity $C(\cW)$ of such a channel, established by Holevo~\cite{holevo98,holevo73b} and Schumacher-Westmoreland~\cite{schumacher97}, is equal to the divergence radius of the channel output,
\begin{align}\label{eq:C-chi-def}
    C(\cW) = \chi(\cW) := \min_{\sigma} \max_x D( \rho_x \| \sigma ),
\end{align}
where the minimisation is over quantum states $\sigma$ and $D(\cdot\|\cdot)$ is the Umegaki relative entropy. The capacity determines how much information can be transmitted reliably over the channel in the asymptotic limit of infinite uses of the channel. The focus has since shifted to the more practical question of finding the fundamental trade-offs between the rate of a code and the incurred error in the transmission. This has been settled in various regimes, for instance by establishing the strong converse property~\cite{winter99} and its exponent~\cite{mosonyi14-2}, as well as the expansion of the maximal coding rate in the small deviation~\cite{tomamicheltan14} and the moderate deviation~\cite{cheng17,chubb17} regimes. Thus, one might think that we completely understand the fundamental limits of coding over cq channels. Not quite! The missing piece in our understanding is the error exponent for communication rates $R$ below but close to the capacity. The maximal (negative) exponent of the decoding error at a rate $R$ is denoted as $E(R)$, so the decoding error behaves as $\exp(-n E(R) + \theta(n))$ for optimal codes at rate $R$. (See Section~\ref{subsec:cq-channel-coding-asymptotic} for a formal definition of $E(R)$.)
For classical channels, the error exponent (or reliability function) can be bounded from both sides as~\cite{shannon59,shannon67,fano61}
\begin{align}
     \sup_{\alpha \in [\frac12, 1)}  \frac{1-\alpha}{\alpha} \big( {\chi}_{\alpha}(\cW) - R \big)\leq E(R) \leq \sup_{\alpha \in (0, 1)}  \frac{1-\alpha}{\alpha} \big( {\chi}_{\alpha}(\cW) - R \big)   , \label{eq:thisform}
\end{align}
where $\chi_\alpha(\cW)$ is the $\alpha$-R\'enyi divergence radius defined similarly to~\eqref{eq:C-chi-def} but in terms of the $\alpha$-R\'enyi divergence~\cite{renyi61}.
The bound of equation~\eqref{eq:thisform} is tight for values of $R$ close to the capacity (above some critical rate) since then the supremum on both sides is taken by some $\alpha \geq \frac12$. 

Let us now consider the setting of coding over cq channels. On the one hand, an \emph{upper} bound on $E(R)$ of the same form as in~\eqref{eq:thisform} was shown by Dalai~\cite{dalai13} (see also~\cite{cheng18} for refinements), and is called the sphere-packing bound. Here the $\alpha$-R\'enyi divergence radius is replaced by $\widebar{\chi}_\alpha(\cW)$, its Petz variant in terms of Petz $\alpha$-R\'enyi divergence:
\begin{align}
    \widebar{\chi}_{\alpha}(\cW) = \min_{\sigma} \max_x \widebar{D}_{\alpha}(\rho_x \| \sigma), \quad \text{where} \quad \widebar{D}_{\alpha}(\rho\|\sigma) = \frac{1}{\alpha-1} \log \tr[\rho^{\alpha}\sigma^{1-\alpha}]\,.
\end{align}
The sphere-packing bound is then written as
\begin{align}
    E(R) \leq \sup_{\alpha \in (0, 1)}  \frac{1-\alpha}{\alpha} \big( \widebar{\chi}_{\alpha}(\cW) - R \big) \label{eq:thisform2}
\end{align}
and is expected to be tight for general states (at  rates above a critical rate). In fact, it is known to be tight for cq channels that only output pure states~\cite{burnashev97} and for some highly symmetric channels~\cite{renes23}. 

On the other hand, the only known \emph{lower} bound on the error exponent for general channels was established by Hayashi in~\cite{hayashi07} and takes the form
\begin{align}
    E(R) \geq \sup_{\alpha \in (0, 1)} (1-\alpha) \big( \max_{P_X} \widebar{D}_{\alpha}(\rho_{XB} \| \rho_X \otimes \rho_B) - R \big), \label{eq:hayashi}
\end{align}
where $\rho_{XB} = \sum_x P(x)\proj{x}_X \otimes \rho_B^x$ with $\rho_B^x = \rho_x$ is the joint state of the channel input and output systems.
However, while this bound is tight for cq channels that output only pure states, it is not equal to the upper bound in~\eqref{eq:thisform2} for general cq channels. In particular, it is not tight for general classical channels, characterised by a set of mutually commuting states $\rho_x$. (This is not obvious. See Appendix~\ref{app:hayashi} for a comparison of the two bounds.)

The impasse for further progress on finding tighter lower bounds for general cq channels seems to be that we are lacking sufficiently good techniques to analyse common decoders for cq channels in such a way that only multiplicative factors arise in the errors bounds. Here, we remedy this by introducing a new class of decoders that are arguably easier to analyse in this setting, and which will eventually lead us to a lower bound of the form~\eqref{eq:thisform} that is tight in the commutative case.

A similar situation as for classical-quantum coding persists with the study of entanglement-assisted channel coding, for which the capacity was first established by Bennett \emph{et al.}~\cite{bennett02}. More recent works have explored the trade-offs between the error and the coding rate in the small deviation~\cite{datta15b}\footnote{For entanglement-assisted channel coding in the small deviation regime currently only lower bounds on the achievable rates could be derived for general channels. We expect them to be tight but a converse is still missing.} and moderate deviation~\cite{ramakrishnan23} regimes. The strong converse property~\cite{gupta13} and strong converse exponent~\cite{gupta13,liyao22} have also recently been found.
However, again the error exponent for rates below the capacity is not yet fully understood. Lower bounds on the error exponent of the form of~\eqref{eq:hayashi} have been established in~\cite{qi18}. We will use our new decoder to find another, incomparable lower bound on this error exponent as well. Our bound is also tight in the case of classical channels.

\medskip
In the following section we fix some notations and introduce certain quantum R\'enyi quantities required in the sequel. Section~\ref{sec:c-q} includes our lower bound on the error exponent for the cq channel coding problem. We also introduce our new decoder in this section. 
Our results on entanglement-assisted channel coding generalise this analysis and are presented in Section~\ref{sec:EA}. We conclude with a discussion in Section~\ref{sec:conclusion}.

\section{Notation and preliminaries}

The logarithm function $\log(\cdot)$ is taken to an arbitrary basis (particularly in base $2$ if we measure capacities in bits per channel use) and $\exp(\cdot)$ is its inverse. We consider the algebra of operators acting on finite-dimensional Hilbert spaces. We use ``$\geq$'' to denote the L\"owner order, i.e., we write $A \geq B$ for two operators $A$ and $B$ to say that $A - B$ is positive semi-definite and $A > B$ to say that $A - B$ is positive definite. A \emph{quantum state} is a positive semi-definite operator with unit trace. We use $1$ to denote the identity operator and $\| \cdot \|$ to denote the operator norm. When considering multi-partite systems, subscripts are used to indicate which subsystems operators act on. We write $A \sim A'$ to indicate that the corresponding Hilbert spaces of two quantum systems are isomorphic. A \emph{quantum channel} is a completely positive trace-preserving (cptp) map from quantum states on some quantum system to quantum states on a potentially different quantum system.

\subsection{R\'enyi divergences}

For two quantum states $\rho$ and $\sigma$ that are not orthogonal, the \emph{sandwiched R\'enyi divergence} of order $\alpha \in \big[\frac12,1\big)$ is defined as~\cite{lennert13,wilde13}
\begin{align}
    \widetilde{D}_{\alpha}(\rho\|\sigma) := \frac{1}{\alpha - 1}\log \widetilde{Q}_{\alpha}(\rho\|\sigma) , \quad
    \widetilde{Q}_{\alpha}(\rho\|\sigma) := \tr \left[ \left( \sigma^{\frac{1-\alpha}{2 \alpha}} \rho \sigma^{\frac{1-\alpha}{2 \alpha}} \right)^{\alpha} \right].
\end{align}
The same definition also applies for $\alpha > 1$, as long as the support of $\rho$ is contained in the support of $\sigma$; but our focus here is only on the regime where $\alpha < 1$. We note that we have a continuous extension at $\alpha=1$:
\begin{align}
    \lim_{\alpha \to 1} \widetilde{D}_{\alpha}(\rho\|\sigma) = D(\rho\|\sigma) := \tr \big[ \rho ( \log \rho - \log \sigma)] ,
\end{align}
which gives Umegaki's relative entropy.
We will also use \emph{measured relative entropies,} defined as
\begin{align}
    \widecheck{D}_{\alpha}(\rho\|\sigma) := \sup_{ \{ \Lambda_j:\, j \}}  \frac{1}{\alpha - 1} \log \left( \sum_j \left( \tr[\rho \Lambda_j] \right)^{\alpha} \left( \tr[\sigma \Lambda_j] \right)^{1-\alpha} \right) ,
\end{align}
where we optimise over sets of orthogonal projectors $\{ \Lambda_j:\, j \}$ partitioning the identity: $\sum_j \Lambda_j=1$. Similarly as before, we define 
\begin{align}
\widecheck{Q}_{\alpha}(\rho\|\sigma) := \exp\big((\alpha -1) \widecheck{D}_{\alpha}(\rho\|\sigma) \big).
\end{align}
It was shown in~\cite[Theorem 4]{berta17} that optimising instead over general positive operator-valued measurements leads to the same quantity, and thus we can restrict ourselves to projective measurements here without loss of generality.

We have $\widebar{D}_{\alpha}(\rho\|\sigma) \geq \widetilde{D}_{\alpha}(\rho\|\sigma)$ as shown in~\cite{datta14} and $\widetilde{D}_{\alpha}(\rho\|\sigma) \geq \widecheck{D}_{\alpha}(\rho\|\sigma)$ by the data-processing inequality for the sandwiched R\'enyi divergence~\cite{frank13,beigi13}. Both inequalities become equalities if and only if the states $\rho$ and $\sigma$ commute~\cite{hiai94,berta17}. It is shown in~\cite{mosonyiogawa13} and~\cite[Corollary 4]{hayashitomamichel16} that for any $\alpha \geq \frac12$,
\begin{align}
    \lim_{n \to \infty} \frac{1}{n} \widecheck{D}_{\alpha}\big(\rho^{\otimes n} \big\|\sigma^{\otimes n} \big) = \widetilde{D}_{\alpha}(\rho\|\sigma) \,. \label{eq:measured-to-sandwich}
\end{align}
We will also use a variational expression of $\widecheck{Q}_{\alpha}(\rho\|\sigma)$ proven in~\cite[Lemma 3]{berta17} stating that for $\alpha \in [\frac12, 1)$, we have
\begin{align}
    \widecheck{Q}_{\alpha}(\rho \| \sigma) 
    &= \inf_{Y > 0} \left( \tr\left[ \rho Y^{\frac{\alpha - 1}{\alpha}} \right] \right)^{\alpha} \left( \tr\left[ \sigma Y \right] \right)^{1-\alpha} \\
    &= \inf_{Y > 0 } \alpha \tr\left[ \rho Y^{\frac{\alpha - 1}{\alpha}} \right]  + (1-\alpha) \tr\left[ \sigma Y \right]  \label{eq:variational} \,.
\end{align}

\subsection{R\'enyi mutual information and divergence radius}

We will express our results in terms of a quantum R\'enyi mutual information. For any bipartite state $\rho_{AB}$, the \emph{sandwiched R\'enyi mutual information} of order $\alpha \in \big[\frac12, 1\big)$ is defined as
\begin{align}
    \widetilde{I}_{\alpha}(A ; B) := \min_{ \sigma_B \geq 0 \atop \tr[\sigma_B] = 1} \widetilde{D}_{\alpha}\big(\rho_{AB} \big\| \rho_A \otimes \sigma_B \big) \,.
\end{align}
It was recently shown that sandwiched R\'enyi mutual information attains operational meanings in some strong converse exponents~\cite{li22} for the range of $\alpha < 1$. 
We can define the same quantity using the measured R\'enyi divergence instead, i.e.,
\begin{align}
    \widecheck{I}_{\alpha}(A ; B ) := \min_{ \sigma_B \geq 0 \atop \tr[\sigma_B] = 1} \widecheck{D}_{\alpha}\big(\rho_{AB} \big\| \rho_A \otimes \sigma_B \big) \,.
\end{align}

Next, we define the \emph{sandwiched and measured R\'enyi divergence radius} of order $\alpha \geq \frac12$ of a set of quantum states $\cW = \{ \rho_x:\, x \in \cX \}$ as
\begin{align}\label{eq:div-radius}
    \widetilde{\chi}_\alpha(\cW) := \min_{\sigma \geq 0 \atop \tr[\sigma] = 1} \max_{x \in \cX}\ \widetilde{D}_{\alpha}(\rho_x \| \sigma) , \quad \textnormal{and} \quad
    \widecheck{\chi}_\alpha(\cW) := \min_{\sigma \geq 0 \atop \tr[\sigma] = 1} \max_{x \in \cX}\ \widecheck{D}_{\alpha}(\rho_x \| \sigma),
\end{align}
respectively. The \emph{sandwiched and measured R\'enyi channel mutual information} of order $\alpha$ are defined by
\begin{align}
    \widetilde{I}_{\alpha}(\cW) := \max_{P_X} \widetilde{I}_{\alpha}(X ; B)_{\rho} \quad \textnormal{and} \quad \widecheck{I}_{\alpha}(\cW) := \max_{P_X} \widecheck{I}_{\alpha}(X ; B)_{\rho}\,,
\end{align}
respectively, where $\rho_{XB} = \sum_x P_X(x) \proj{x}_X \otimes \rho_B^x$ with $\rho_B^x = \rho_x$.
These quantities have several alternative expressions and have been studied extensively, for example in~\cite{mosonyi14-2}, where an operational meaning in terms of the strong converse exponent in cq channel coding was found when $\alpha > 1$. In particular, it is shown in~\cite[Proposition 4.2]{mosonyi14-2} that
\begin{align}
    \widetilde{\chi}_\alpha(\cW) = \widetilde{I}_{\alpha}(\cW)\,.  \label{eq:equivalent}
\end{align}
We can prove a similar equivalence also for measured quantities.

\begin{lemma} \label{lem:measured-radius}
    For any $\alpha \geq \frac12$ it holds that $\widecheck{\chi}_\alpha(\cW) = \widecheck{I}_{\alpha}(\cW)$.
\end{lemma}

\begin{proof}
    Let us first consider the case $\alpha \in [\frac12, 1)$ that is of most interest here. We note that
    \begin{align}
        \max_{P_X} \widecheck{I}_{\alpha}(X ; B)_{\rho} = \frac1{\alpha -1} \log \left( \min_{P} \max_{\sigma \geq 0 \atop \tr[\sigma] = 1} \sum_x P(x) \widecheck{Q}_{\alpha}(\rho_x \| \sigma) \right). 
    \end{align}
    Since the expression $\sum_x P(x) \widecheck{Q}_{\alpha}(\rho_x \| \sigma)$ is linear in $P$ and concave in $\sigma$, we apply Sion's minimax theorem~\cite{sion58} to exchange the optimisations. Finally, we observe that 
    \begin{align}
        \min_P\ \sum_x P(x) \widecheck{Q}_{\alpha}(\rho_x \| \sigma) = \min_x \widecheck{Q}_{\alpha}(\rho_x \| \sigma).
    \end{align}
    This yields the desired result.    
    The cases $\alpha = 1$ and $\alpha > 1$ follow similarly.
\end{proof}

For general quantum channels $\cN$ from $A$ to $B$, the \emph{sandwiched R\'enyi channel mutual information} of order $\alpha$ is defined as
\begin{align}
    \widetilde{I}_{\alpha}(\cN) := \max_{\rho_A} \widetilde{I}_{\alpha}(A' ; B)_{\tau}, \quad \textnormal{where} \quad \tau_{A'B} = \mathcal{I} \otimes \cN \left( \proj{\rho}_{A'A} \right)\,,
\end{align}
and $|\rho\rangle_{A'A}$ is any purification of $\rho_A$ on an isomorphic system $A' \sim A$. Similarly, the \emph{measured R\'enyi channel mutual information} of order $\alpha$ is defined as
\begin{align}
    \widecheck{I}_{\alpha}(\cN) := \max_{\rho_A} \widecheck{I}_{\alpha}(A' ; B)_{\tau} \,.
\end{align}

\section{Error exponent for classical-quantum channel coding}\label{sec:c-q}

We consider a classical-quantum channel  characterised by an ensemble $\cW = \{ \rho_x:\, x \in \cX \}$ of quantum states on some fixed Hilbert space $\cH$. Here, $\cX$ is the input alphabet. We are interested in sending a (uniformly distributed) message $m \in \cM = \{1, 2, \ldots, M \}$ through this channel. A \emph{channel code} $\cC$ of size $|\cC| = M$ is comprised of an encoder function, $e: \cM \to \cX$ and positive operator-valued measurement $\{ \Pi_m :\, m \in \cM\}$, the decoder. The \emph{average probability of error} for such a code is given by
\begin{align}
    p_{\textnormal{err}}(\cC, \cW) = 1 - \frac{1}{M} \sum_{m}  \tr \big[ \rho_{e(m)} \Pi_m \big] = \frac{1}{M} \sum_{m \neq m'} \tr \big[ \rho_{e(m)} \Pi_{m'} \big] \,.
\end{align}

\subsection{One-shot bound}

The following theorem is a one-shot bound on the average error probability based on random coding.

\begin{theorem}
    \label{thm:one-shot-error}
    Let $M \in \mathbb{N}$, $P_X$ be a distribution on $\cX$, and $\alpha \in \big[\frac12, 1\big)$. There exists a code $\cC$ over $\cW$ with $|\cC| = M$ satisfying
    \begin{align}
        p_{\textnormal{err}}(\cC, \cW) & \leq (M-1)^{\frac{1-\alpha}{\alpha}} \max_{\sigma \geq 0 \atop \tr[\sigma] = 1} \left( \E_{x}\, \widecheck{Q}_{\alpha}( \rho_x \| \sigma) \right) ^{\frac{1}{\alpha}} \label{eq:cq-bound-1}\\
        &\leq \exp \left( \frac{1-\alpha}{\alpha} \big( \log M - \widecheck{I}_{\alpha}(X ; B)_{\rho} \big) \right),\label{eq:cq-bound-2}
    \end{align}
    where the expectation is with respect to distribution $P_X$ and $\rho_{XB} = \sum_x P_X(x) \proj{x}_X \otimes \rho_B^x$ with $\rho_B^x = \rho_x$. In particular, there exists a code satisfying
    \begin{align}
        p_{\textnormal{err}}(\cC, \cW) \leq \exp \left( \frac{1-\alpha}{\alpha} \big( \log M - \widecheck{\chi}_{\alpha}(\cW) \big) \right) .\label{eq:bound-3}
    \end{align}
\end{theorem}

The expression has no fudge terms and, in case where all $\rho_x$'s commute, we can replace the measured quantities with the usual R\'enyi quantities. Also note that the measurement based on which $\widecheck{Q}_{\alpha}( \rho_x \| \sigma)$ is defined, depends on the input $x$. 
Thus, we cannot simply think of this measurement as the classical error exponent for a channel where a fixed measurement is appended as a post-processing.

The rest of this subsection is devoted to the proof of this theorem. Later, in Subsection~\ref{subsec:cq-channel-coding-asymptotic} we compute the asymptotic bound derived from this theorem.
The proof of Theorem~\ref{thm:one-shot-error} is based on a new class of decoders that is more suitable for the analysis of error at least for the error exponent problem. To motivate the definition of these decoders we first focus on the channel coding problem in the fully classical setting, and present a tight analysis of the error. Then, we explain how our new decoder can be used to generalise this analysis to the quantum case.

\subsubsection{Classical channels}

Suppose that the output of the channel $\cW$ is classical meaning that all states $\rho_x$, $x\in \mathcal X$, commute. As usual we employ a random encoding, and given an input distribution $P_X$, choose codewords $x_1, \dots, x_M$ independently at random according to $P_X$. For the decoder we fix some parameter $\alpha\in [1/2, 1)$ and choose 
\begin{align}\label{eq:def-decoder-classic}
\Pi_m = \frac{\rho_{x_m}^\alpha}{\sum_{m'}\rho_{x_{m'}}^\alpha }.
\end{align}
We note that $\Pi_m$ is a positive semidefinite matrix and $\sum_m \Pi_m=1$, so these operators form a valid measurement. We emphasise that here $\rho_x$'s commute and there is no ambiguity in the definition of $\Pi_m$. Such a decoder has already been studied in the classical setting in~\cite{liu2017alpha} for the error exponent problem. See also~\cite{scarlett2015likelihood} for similar results. Nevertheless, here we deviate from the methods of~\cite{liu2017alpha, scarlett2015likelihood} in order to make our analysis more amenable to a quantum generalisation.

Writing down the expectation of the average probably of error, we compute
\begin{align}
    \E_{x_1, \dots, x_M}[p_{\textnormal{err}}] & = \frac{1}{M}\E_{x_1, \dots, x_M} \left[\sum_{m=1}^M \tr\left[ 
 \rho_{x_m} \left(\sum_{m'\neq m}\Pi_{m'}\right) \right]\right]\\
  & = \E_{x_1, \dots, x_M} \tr\left[ 
 \rho_{x_1} \left(\sum_{m'\neq 1}\Pi_{m'}\right) \right]\\
   & = \E_{x_1} \tr\left[ 
 \rho_{x_1} \left(\E_{x_2, \dots, x_M}\sum_{m'\neq 1}\Pi_{m'}\right) \right].
\end{align}
Here, in the second line we use the symmetry of the problem to consider only the error associated to message $m=1$, and in the third line we use the linearity of expectation and the fact that the codewords are chosen independently. Next, we note that by definition $\sum_{m'\neq 1}\Pi_{m'}\leq 
1$ and taking the expectation we still have $\E_{x_2, \dots, x_M}\sum_{m'\neq 1}\Pi_{m'}\leq 1$. Then, letting $r=\frac{1-\alpha}{\alpha}$, since $\alpha\in [1/2, 1)$, we have $r\in(0, 1]$ and
\begin{align}\label{eq:raise-power-r-classic}
\E_{x_2, \dots, x_M}\sum_{m'\neq 1}\Pi_{m'}\leq \left(\E_{x_2, \dots, x_M}\sum_{m'\neq 1}\Pi_{m'}\right)^r.
\end{align}
This yields 
\begin{align}
    \E_{x_1, \dots, x_M}[p_{\textnormal{err}}]
   & \leq \E_{x_1} \tr\left[ 
 \rho_{x_1} \left(\E_{x_2, \dots, x_M}\sum_{m'\neq 1}\Pi_{m'}\right)^r \right].
\end{align}
Using the definition of the decoder, and the fact that $\rho_x$'s commute, we have
\begin{align}\label{eq:denom-eliminate-classic}
    \Pi_m = \frac{\rho_{x_m}^\alpha}{\sum_{m'\neq m}\rho_{x'}^\alpha }\leq \frac{\rho_{x_m}^\alpha}{\rho_{x_1}^\alpha}.
\end{align}
Hence, the monotonicity of $t\mapsto t^r$ implies 
\begin{align}
    \E_{x_1, \dots, x_M}[p_{\textnormal{err}}]
   & \leq \E_{x_1} \tr\left[ 
 \rho_{x_1} \left(\E_{x_2, \dots, x_M}\sum_{m'\neq 1}\frac{\rho_{x_m}^\alpha}{\rho_{x_1}^\alpha}\right)^r \right]\\
 & = \E_{x_1} \tr\left[ 
 \rho_{x_1} \left((M-1) \frac{\E_{x_2} \rho_{x_2}^\alpha}{\rho_{x_1}^\alpha}\right)^r \right]\\
 & = (M-1)^r \E_{x_1} \tr\left[ 
 \rho_{x_1}^{1-r\alpha} \left( \E_{x_2}\rho_{x_2}^\alpha\right)^r \right]\label{eq:comp-power-classic}\\
 & = (M-1)^r  \tr\left[ \left(
 \E_{x_1} \rho_{x_1}^{\alpha} \right)\left( \E_{x_2}\rho_{x_2}^\alpha\right)^r \right]\\
 & = (M-1)^r  \tr\left[ \left(
 \E_{x} \rho_{x}^{\alpha} \right)^{1+r} \right]. \label{eq:fin-bound-classic}
\end{align}
Here, in the third line we use the fact that $\rho_x$'s commute and in the fourth line we use the definition of $r$ to conclude  that $1-r\alpha=\alpha$. It is readily verified that~\eqref{eq:fin-bound-classic} is equivalent to the bounds~\eqref{eq:cq-bound-1} of Theorem~\ref{thm:one-shot-error} for classical channels by performing the optimisation over $\sigma$ in the latter bound.

\subsubsection{Classical-quantum channels}

To generalise the above analysis to the quantum case, we need to verify the validity of the four key steps in~\eqref{eq:def-decoder-classic},~\eqref{eq:raise-power-r-classic},~\eqref{eq:denom-eliminate-classic} and~\eqref{eq:comp-power-classic} when $\rho_x$'s do not commute.
First, a standard way of defining the measurement operators in~\eqref{eq:def-decoder-classic} is via the \emph{matrix division:}
\begin{align}
    \altfrac{A}{B} = B^{-\frac 12}AB^{-\frac 12},
\end{align}
which gives the pretty-good measurement operators
\begin{align}
\Pi_m'=\left(\sum_{m'}\rho_{x_{m'}}^\alpha\right)^{-\frac 12} \rho_{x_{m}}^\alpha \left(\sum_{m'}\rho_{x_{m'}}^\alpha\right)^{-\frac 12}.
\end{align}
Inequality~\eqref{eq:raise-power-r-classic} holds for these and in fact for any measurement operators. Nevertheless, it is not hard to come up with examples of quantum states for which the above matrix division violates~\eqref{eq:denom-eliminate-classic}.  

Our idea is to use a non-standard definition of matrix quotient that satisfies~\eqref{eq:denom-eliminate-classic}. To this end, for any matrix $A$ and a positive definite matrix $B > 0$ we define
\begin{align}\label{eq:int-matrix-quotient}
    \frac{A}{B} := \int_0^{\infty} \frac{1}{\lambda+B} A \frac{1}{\lambda+B} \, \d \lambda.
\end{align}
We can also understand the above quotient as the Fr\'echet derivative $D\ln(B)[A]$ of the logarithm. We note that $\int_{0}^{\infty} \d\lambda\, (\lambda+b)^{-2} = b^{-1}$, so the above integral evaluates to $A B^{-1}$ when $A, B$ commute.
The following lemma shows that this matrix division satisfies~\eqref{eq:denom-eliminate-classic}. 

\begin{lemma} \label{lem:quotient}
    For any matrices $A \geq 0$ and $B > 0$, we have 
    \begin{align}
        \frac{A}{A+B} \leq \frac{A}{B} \,.
    \end{align}
\end{lemma}

\begin{proof}
    We first observe that
    \begin{align}
        (\lambda + A + B) A^{-1} (\lambda + A + B) 
        &= (\lambda + B) A^{-1} (\lambda + B) + A + 2 (\lambda + B) \\
        &\geq (\lambda + B) A^{-1} (\lambda + B).
    \end{align}
    Hence, by the operator anti-monotonicity of $t\mapsto t^{-1}$, we get
    \begin{align}
        \frac{1}{\lambda + B + A} A \frac{1}{\lambda + A + B} \leq \frac{1}{\lambda + B} A \frac{1}{\lambda + B}.
    \end{align}
    Integrating this inequality, we obtain the desired result.
\end{proof}

Using this lemma, the above classical argument goes through up to the point where we apply equation~\eqref{eq:fin-bound-classic}. This equation, however, is far from true in the non-commutative case, so we need an extra ingredient to complete the analysis. 

In the non-commutative case, we may generalise the measurement by first fixing some operators $Y_x> 0$ and then defining
\begin{align}\label{eq:def-Pi-Y}
    \Pi_m= \frac{Y_{x_m}}{\sum_{m'} Y_{x_{m'}}}.
\end{align}
Of course, letting $Y_x$ to be $\rho_x^\alpha$, we recover the previous measurement, yet we leave $Y_x$'s arbitrary for now.\footnote{Letting $Y_{x} = \rho_{x}^\alpha=\rho_x$ for $\alpha=1$, we obtain the \emph{universal recovery map} of~\cite{junge2018universal}. To understand this connection see~\cite[Lemma 3.4]{sutter16}. } Then, repeating the above argument for the measurement operators~\eqref{eq:def-Pi-Y} and using Lemma~\ref{lem:quotient} we arrive at
\begin{align}
    \E_{x_1, \dots, x_M}[p_{\textnormal{err}}]
 \leq (M-1)^r \E_{x_1} \tr\left[ 
 \rho_{x_1} \left( \frac{\E_{x_2} Y_{x_2}}{Y_{x_1}}\right)^r \right].
\end{align}
Now we note that $\frac{\E_{x_2} Y_{x_2}}{Y_{x_1}}\leq \|\E_{x_2} Y_{x_2}\|\cdot \frac{1}{Y_{x_1}}$.
Then, using the operator monotonicity of $t\mapsto t^r$ and optimising over $Y_x$'s,  we conclude that
\begin{align}
    \E_{x_1, \dots, x_M}[p_{\textnormal{err}}]
   & \leq (M-1)^r\inf_{Y_x>0} \E_{x_1} \tr\left[ 
 \rho_{x_1} Y_{x_1}^{-r} \right]\cdot \|\E_{x_2} Y_{x_2}\|^r\\
 &= (M-1)^r\left(\inf_{Y_x>0} \left(\E_{x_1} \tr\left[ 
 \rho_{x_1} Y_{x_1}^{-r} \right]\right)^\alpha\cdot \|\E_{x_2} Y_{x_2}\|^{1-\alpha}\right)^{\frac{1}{\alpha}}.
\end{align}
Next, applying the weighted arithmetic-geometric mean inequality (Young's inequality) we arrive at
\begin{align}
   \E_{x_1, \dots, x_M}[p_{\textnormal{err}}]
   & \leq (M-1)^r\left(\inf_{Y_x>0} \alpha\E_{x} \tr\left[ 
 \rho_{x_1} Y_{x}^{-r} \right]+(1-\alpha) \|\E_{x} Y_{x}\|\right)^{\frac{1}{\alpha}}\\
 & \leq (M-1)^r\left(\inf_{Y_x>0} \max_{\sigma_B} \E_x\big[ \alpha \tr\left[ 
 \rho_{x} Y_{x}^{-r} \right]+(1-\alpha) \tr[\sigma Y_{x}]\big]\right)^{\frac{1}{\alpha}}\\
 & = (M-1)^r\left( \max_{\sigma_B} \inf_{Y_x>0} \E_x\big[ \alpha \tr\left[ 
 \rho_{x} Y_{x}^{-r} \right]+(1-\alpha) \tr[\sigma Y_{x}]\big]\right)^{\frac{1}{\alpha}}\label{eq:cq-bound-Y}\\
 & = (M-1)^r\left( \max_{\sigma_B}  \E_x  \widecheck{Q}_{\alpha}(\rho_x \| \sigma) \right)^{\frac{1}{\alpha}}\,. \label{eq:cq-bound-Y-2}
\end{align}
Here, in~\eqref{eq:cq-bound-Y} we use the operator convexity of $t\mapsto t^{-r}$ to apply Sion's minimax theorem~\cite{sion58}. Also,~\eqref{eq:cq-bound-Y-2} follows from~\eqref{eq:variational}. This gives~\eqref{eq:cq-bound-1}. To prove~\eqref{eq:cq-bound-2} we need to verify that  
\begin{align}
        \mathbb E_x\widecheck{Q}_{\alpha}(\rho_x \| \sigma)  =  \widecheck{Q}_{\alpha}(\rho_{XB} \| P_X\otimes \sigma). 
         \label{eq:variational-cq-2}
\end{align}
This holds since in the definition of $\widecheck{Q}_{\alpha}(\rho_{XB} \| P_X\otimes \sigma)$ we can with no loss of generality restrict to measurements that respect the cq structure. The point is that given a measurement  $\{\Lambda_j:\, j\}$, we have $\tr(\rho_{XB}\Lambda_j) = \E_x \tr(\rho_x \Lambda_{j}^x)$ where $\Lambda_j^x = \bra{x} \Lambda_j \ket x$ and $\{\Lambda_j^x:\, j\}$ for any $x$ forms a measurement.


\subsubsection{Further properties of the new decoder}

In Appendix~\ref{app:Hayashi's-bound-alt} we show that the new decoder can also be used to derive Hayashi's bound, Eq.~\eqref{eq:hayashi}. The main ingredient is the following lemma, which establishes Cheng's tight one-shot bound on the error probability~\cite[Lemma 17]{Cheng22} for the new matrix quotient. 

\begin{lemma}\label{lem:int-measurement-error-bound}
For positive semidefinite operators $A, B$ acting on a Hilbert space $\cH$ we have 
\begin{align}\label{eq:int-measurement-error-bound}
\tr\left[A \frac{B}{A+B}\right] \leq \frac12 \big( \tr[A+ B] - \|A-B\|_1 \big).
\end{align}
\end{lemma}

Cheng's bound circumvents the famous Hayashi-Nagaoka operator inequality~\cite[Lemma 2]{hayashi03} and can be used to derive most known asymptotic achievability bounds on the error exponent as well as second-order or moderate-deviation expansions~\cite{Cheng22}. As a consequence, we establish that our decoder performs equally well in all these regimes. The proof follows~\cite{Cheng22} closely, and we present it in Appendix~\ref{app:Hayashi's-bound-alt} for completeness.

\subsection{Asymptotic bound}\label{subsec:cq-channel-coding-asymptotic}

Our second main result is an expression for the asymptotic error exponent. The error exponent of the classical-quantum channel coding problem is formally defined, for each rate $R > 0$, as
\begin{align}
    E(R, \cW) := \sup_{ \{\cC_n:\, n\} } \left\{ \liminf_{n\to\infty} -\frac{\log p_{\textnormal{err}}(\cC_n, \cW^{\otimes n})}{n}  :\,\, \liminf_{n\to\infty} \frac{\log |\cC_n|}{n} \geq R \right\} ,
\end{align}
where the optimisation is over sequences of codes $\{\cC_n:\, n\geq 1\}$ and $|\cC_n|$ is the number of codewords in the code.

\begin{theorem}\label{thm:asymptotic}
    Let $R > 0$. Then,
    \begin{align}
       & E(R, \cW) \geq \sup_{\alpha \in [\frac12, 1)} \left\{  \frac{1-\alpha}{\alpha} \big( \widetilde{\chi}_{\alpha}(\cW) - R \big) \right\} \,. 
\end{align}
\end{theorem}


We know that $\widetilde{\chi}_\alpha(\cW) \leq \widebar{\chi}_\alpha(\cW)$ since the sandwiched R\'enyi divergence is a lower bound to the Petz R\'enyi divergence for $\alpha\in [\frac 12, 1)$. Moreover, our bound is not comparable with~\eqref{eq:hayashi} for general states and neither bound is tight for pure states. However, our bound is tight (for rates above the critical rate) at least in the following situations:
\begin{itemize}
    \item[-] when $\{ \rho_x \}_x$ mutually commute, i.e., in the classical case;
    \item[-] for symmetric channels where $\rho_x = U_x \rho_0 U_x^{\dagger}$ for some unitary 1-design $\{ U_x :\, x\}$. (This recovers a result in~\cite{renes23}.)
\end{itemize}
In both cases this is because the optimizer $\sigma$ in~\eqref{eq:div-radius} can be chosen to commute with all $\rho_x$.

The proof of Theorem~\ref{thm:asymptotic} is based on the following lemma which shows that the sandwiched R\'enyi quantities emerge in an asymptotic limit of the measured variants, similar to~\eqref{eq:measured-to-sandwich}.

\begin{lemma} \label{lem:measured-limit} 
    Let $\rho_{AB}$ be a bipartite state and $\tau_A$ be arbitrary. For any $\alpha \geq \frac12$ we have
    \begin{align}
        \lim_{n \to \infty} \frac{1}{n} \widecheck{I}_{\alpha}\big( A^n ; B^n \big)_{\rho^{\otimes n}} = \widetilde{I}_{\alpha}(A; B)_{\rho} \,. 
    \end{align}
\end{lemma}

\begin{proof}
    The direction `$\leq$' is a consequence of the data-processing inequality together with the fact that $\widetilde{I}_{\alpha}(\cdot; \cdot)$ is additive under tensor products~\cite[Lemma 7]{hayashitomamichel16}. To prove the reverse direction `$\geq$' first observe that by applying the twirling map with respect to permutations of the $n$ systems and using the data-processing inequality we find that the function
    \begin{align}
        \sigma_{B^n} \mapsto \frac{1}{n} \widecheck{D}_{\alpha}\big( \rho_{AB}^{\otimes n} \big\| \rho_{A}^{\otimes n} \otimes \sigma_{B^n} \big) 
    \end{align}
    has a minimizer $\sigma_{B^n}$ that is invariant under permutations of subsystems. Let us call this permutation invariant minimizer $\omega_{B^n}$. Note that the number of distinct eigenvalues of $\omega_{B^n}$ satisfies $|\spec(\omega_{B^n})| = \poly(n)$; see, e.g., ~\cite[Lemma~2.4]{berta21}.
    Next, observe that applying the pinching map $\cP_{\rho_{A}^{\otimes n} \otimes \omega_{B^n}}(\cdot)$ on the first argument forces the resulting states to commute and is thus equivalent to a measurement in the joint eigenbasis. Therefore,  
    \begin{align}
    \widecheck{I}_{\alpha}\big( A^n ; B^n \big)_{\rho^{\otimes n}}
        & \geq \widetilde{D}_{\alpha}\left( \cP_{\rho_{A}^{\otimes n} \otimes \tau_{B^n}}\big(\rho_{AB}^{\otimes n}\big) \middle\| \rho_{A}^{\otimes n} \otimes \omega_{B^n} \right) \\
        & \geq \widetilde{D}_{\alpha} \big( \rho_{AB}^{\otimes n} \big\| \rho_{A}^{\otimes n} \otimes \omega_{B^n} \big) - \log |\spec(\omega_{B^n})| \label{eq:pinching-ineq} \\
        & \geq \widetilde{I}\big( A^n ; B^n \big)_{\rho^{\otimes n}} - \log \poly(n) \\
        & = n \widetilde{I}( A ; B)_{\rho} - \log \poly(n) ,
    \end{align}
    where~\eqref{eq:pinching-ineq} is established using~\cite[Lemma 3]{hayashitomamichel16} and the final equality is due to the additivity of $\widetilde{I}_{\alpha}(\cdot; \cdot)$ under tensor products. The proof concludes after dividing by $n$ and taking the limit $n \to +\infty$ on both sides.
\end{proof}

\begin{proof}[Proof of Theorem~\ref{thm:asymptotic}]
    We apply Theorem~\ref{thm:one-shot-error} with $P_{X^n}=P_X^n$ being i.i.d according to some distribution $P_X$ that is still to be determined. We also let $M = \exp(nR)$. This yields
    a code satisfying
    \begin{align}
        -\frac{1}{n} \log p_{\textnormal{err}}(\cC_n, \cW^{\otimes n}) \geq \frac{1-\alpha}{\alpha} \left( \frac{1}{n} \widecheck{I}_{\alpha}(X^n; B^n)_{\rho^{\otimes n}} - R \right).
    \end{align}
    Taking the limit $n \to +\infty$ and using Lemma~\ref{lem:measured-limit} lead us to the bound
    \begin{align}
        \lim_{n \to \infty} -\frac{1}{n} \log p_{\textnormal{err}}(\cC_n, \cW^{\otimes n}) \geq \frac{1-\alpha}{\alpha} \left( \widetilde{I}_{\alpha}(X; B)_{\rho} - R \right) \,.
    \end{align}
    The desired result then follows by optimising over the choice of the distribution $P_X$ and $\alpha \in \big[\frac12, 1\big)$ and using~\eqref{eq:equivalent}.
\end{proof}

\section{Error exponents for entanglement-assisted channel coding}\label{sec:EA}

Entanglement-assisted communication can be seen as a natural fully quantum generalisation of classical channel coding~\cite{bennett02}. A quantum channel $\cN$ is a completely positive trace-preserving (cptp) map from quantum states on an input space $A$ (Alice) to quantum states on an output space $B$ (Bob). We are interested in sending a uniformly distributed message $m \in \cM$ over this channel. An \emph{entanglement-assisted channel code} $\cC$ for a quantum channel is comprised of a state $\tau_{\bar{A}\bar{B}}$ shared between Alice and Bob, an encoder $\{ \cE_m : m \in \cM \}$ where $\cE_m$ are cptp maps from $\bar{A}$ to $A$ and a decoder $\{ \Pi_m : m \in \cM \}$ forming a positive operator-valued measurement on $\bar{B}B$. The average probability of error is then given by
\begin{align}
    p_{\textnormal{err}}(\cC, \cN) = \frac{1}{M} \sum_{m \neq m'} \tr \left[ \Pi_{m'} \left( (\cN \circ \cE_m) \otimes \mathcal{I} \right) \phi_{\bar{A}\bar{B}} \right].
\end{align}

\subsection{One-shot bound}

We first derive the following one-shot bound in terms of the measured mutual information.

\begin{theorem}
    Let $M \in \mathbb{N}$, $\rho_{A}$ be a quantum state and $\alpha \in \big[\frac12, 1\big)$. There exists an entanglement-assisted code $\cC$ over $\cN$ satisfying
    \begin{align}
        p_{\textnormal{err}}(\cC, \cN) \leq \exp\left( \frac{1-\alpha}{\alpha} \Big( \log M - \widecheck{I}_{\alpha}(A'; B)_{\rho} \Big) \right)\,,
    \end{align}
    where $\rho_{A'B} = \mathcal{I} \otimes \cN\big( \proj{\rho}_{A'A} \big)$ and $|\rho\rangle_{A'A}$ is any purification of $\rho_A$ on $A' \sim A$.
\end{theorem}

\begin{proof}
    Our one-shot bound uses a coding strategy based on position-based coding~\cite{anshu19b}. We set
    \begin{align}
        \tau_{\bar{A}\bar{B}} = \rho_{AA'}^{\otimes M}, \quad \textnormal{where} \quad \bar{A} = {A}_1 \otimes \cdots \otimes {A}_M \sim A^{\otimes M} \quad \textnormal{and} \quad \bar{B} = {B}_1 \otimes \cdots \otimes {B}_M \sim A^{\otimes M} \,.
    \end{align}
    The encoder $\cE_m$ simply prepares the $m$-th subsystem $A_m$ as the channel input $A$. The decoder is parameterised by a positive definite operator $Y_{AB}$ and
    \begin{align}
        \Pi_m = \frac{Y_{\bar{B} B}^m}{\sum_{m'} Y_{\bar{B} B}^{m'}}, \quad
        \textnormal{where} \quad
        Y_{\bar{B} B}^m = 1_{B_1} \otimes \cdots \otimes 1_{B_{m-1}} \otimes Y_{B_m B} \otimes 1_{B_{m+1}} \cdots\otimes 1_{B_M}.
    \end{align}
    Here, $Y_{B_m B}$ is the embedding of $Y_{AB}$ into $\bar{B}B$ that acts non-trivially only on the subsystems $B_m$ and $B$. Setting $r = \frac{\alpha-1}{\alpha} \in (0, 1]$ and leveraging on the symmetry of the problem, we can then analyse the error as
    \begin{align}
        p_{\textnormal{err}} &= \tr \left[ \rho_{{B}_1 B} \otimes \rho_{{B}_2} \otimes \cdots 
 \rho_{{B}_M} \left( \frac{\sum_{m' \neq 1} Y_{\bar{B} B}^m}{\sum_{m'} Y_{\bar{B} B}^{m'}} \right) \right] \\
        &= \tr \left[ \rho_{{B}_1 B} \left( \tr_{B_2 \ldots B_M} \left[ 1_{B_1B} \otimes \rho_{{B}_2} \otimes \cdots 
 \rho_{{B}_M} \left( \frac{\sum_{m' \neq 1} Y_{\bar{B} B}^m}{\sum_{m'} Y_{\bar{B} B}^{m'}} \right) \right] \right) \right] \\
        &\leq \tr \left[ \rho_{{B}_1 B} \left( \tr_{B_2 \ldots B_M} \left[ 1_{B_1B} \otimes \rho_{{B}_2} \otimes \cdots 
 \rho_{{B}_M} \left( \frac{\sum_{m' \neq 1} Y_{\bar{B} B}^m}{\sum_{m'} Y_{\bar{B} B}^{m'}} \right) \right] \right)^r \right] \\
        &\leq \tr \left[ \rho_{{B}_1 B} \left( \tr_{B_2 \ldots B_M} \left[ 1_{B_1B} \otimes \rho_{{B}_2} \otimes \cdots 
 \rho_{{B}_M} \left( \frac{\sum_{m' \neq 1} Y_{\bar{B} B}^m}{Y_{\bar{B} B}^{1}} \right) \right] \right)^r \right] \\
 &= \tr \left[ \rho_{{B}_1 B} \left( \sum_{m' \neq 1} \tr_{B_{m'}} \left[ 1_{B_1B} \otimes \rho_{{B}_{m'}} \frac{ 1_{B_1} \otimes Y_{B_{m'} B}}{Y_{B_1 B} \otimes 1_{B_{m'}}} \right] \right)^r \right] \\
 &= (M-1)^r \tr \left[ \rho_{B_1 B} \left(  \frac{ 1_{B_1} \otimes \tr_{A} \left[ \rho_{A} Y_{A B} \right] }{Y_{B_1 B}} \right)^r \right] \\
 &\leq (M-1)^r \tr \left[ \rho_{A B} Y_{AB}^{-r} \right] \cdot \| \tr_{A} \left[ \rho_{A} Y_{A B} \right] \|^r \,.
    \end{align}
    In the above we use the fact that $$\tr_{B_2 \ldots B_M} \Big[ \rho_{{B}_2} \otimes \cdots 
 \rho_{{B}_M} \frac{\sum_{m' \neq 1} Y_{\bar{B} B}^m}{\sum_{m'} Y_{\bar{B} B}^{m'}} \Big] \leq 1_{B_1B},$$ 
 to verify the first inequality, and Lemma~\ref{lem:quotient} and the operator monotonicity of $t\mapsto t^r$ to verify the second inequality. Now minimising over $Y_{AB}$ and using the weighted arithmetic-geometric mean inequality (Young's inequality), we arrive at
 \begin{align}
     p_{\textnormal{err}} &\leq (M-1)^r \inf_{Y_{AB} > 0} \tr \left[ \rho_{A B} Y_{AB}^{-r} \right] \cdot \| \tr_{A} \left[ \rho_{A} Y_{A B} \right] \|^r \\
     &= (M-1)^r \left( \inf_{Y_{AB} > 0} \left( \tr \left[ \rho_{A B} Y_{AB}^{-r} \right] \right)^{\alpha} \cdot \| \tr_{A} \left[ \rho_{A} Y_{A B} \right] \|^{1-\alpha} \right)^{\frac{1}{\alpha}} \\
     &\leq (M-1)^r \left( \inf_{Y_{AB} > 0} \alpha \tr \left[ \rho_{A B} Y_{AB}^{-r} \right] + (1-\alpha) \| \tr_{A} \left[ \rho_{A} Y_{A B} \right] \| \right)^{\frac1{\alpha}} \\
     &= (M-1)^r \left( \inf_{Y_{AB} > 0} \max_{\sigma_B \geq 0 \atop \tr[\sigma_B] = 1} \alpha \tr \left[ \rho_{A B} Y_{AB}^{-r} \right] + (1-\alpha) \tr \left[ \rho_A \otimes \sigma_B Y_{A B} \right] \right)^{\frac1{\alpha}} \,.
 \end{align}
Now looking at the inner optimisation we note that Sion's minimax theorem~\cite{sion58} applies since the functional is linear in $\sigma$ and convex in $Y_{AB}$ due to the operator convexity of $t\mapsto t^{-r}$. Thus, we have
 \begin{align}
     p_{\textnormal{err}} &\leq (M-1)^r \left(  \max_{\sigma_B \geq 0 \atop \tr[\sigma_B] = 1} \inf_{Y_{AB} > 0} \alpha \tr \left[ \rho_{A B} Y_{AB}^{-r} \right] + (1-\alpha) \tr \left[ \rho_A \otimes \sigma_B Y_{A B} \right] \right)^{\frac1{\alpha}} \\
     &= (M-1)^r \left( \max_{\sigma_B \geq 0 \atop \tr[\sigma_B] = 1} \widecheck{Q}_{\alpha}(\rho_{AB} \| \rho_A \otimes \sigma_B) \right)^{\frac1{\alpha}} \,,
 \end{align}
 where we used the variational formula in~\eqref{eq:variational} in the last step. Since $M-1 \leq M$ this yields the desired bound after some rewriting.
\end{proof}

\subsection{Asymptotic bound}

The error exponent for entanglement-assisted channel coding is formally defined, for each rate $R > 0$, as
\begin{align}
    E(R, \cN) := \sup_{ \{\cC_n:\, n\} } \left\{ \liminf_{n\to\infty} -\frac{\log p_{\textnormal{err}}(\cC_n, \cN^{\otimes n})}{n}  :\,\, \liminf_{n\to\infty} \frac{\log |\cC_n|}{n} \geq R \right\} ,
\end{align}
where the optimisation is over sequences of codes $\{\cC_n:\, n\geq 1\}$. Lemma~\ref{lem:measured-limit} together with an argument analogous to the proof of Theorem~\ref{thm:asymptotic} then directly yields a lower bound on the error exponent.

\begin{theorem}
    Let $R > 0$. Then,
    \begin{align}
       & E(R, \cN) \geq \sup_{\alpha \in [\frac12, 1)} \left\{  \frac{1-\alpha}{\alpha} \big( \widetilde{I}_{\alpha}(\cN) - R \big) \right\} \,.
\end{align}
\end{theorem}

\section{Conclusion}\label{sec:conclusion}

We have introduced a new type of pretty good measurement decoder and showed how it can be used to give lower bounds on the error exponents for classical-quantum and entanglement-assisted channel coding that are tight in the classical commutative case. We believe that our decoder and techniques will find applications beyond these two problems.

An example of such a problem is source coding with quantum side information, where a memoryless source produces classical-quantum states $\rho_{XB}$ where the classical source $X$ needs to be compressed so that it can be recovered using the quantum side information $B$. This problem was first studied in the asymptotic setting~\cite{devetak03} and the finite resource trade-offs have been investigated in~\cite{tomamichel12} as well as~\cite{cheng21} and~\cite{renes23}. The latter two works establish that the error exponent of this problem for rates $R$ above but close to the conditional entropy is equal to
\begin{align}
    \sup_{\alpha \in [\frac12, 1)} \frac{1-\alpha}{\alpha} \left( R - \widebar{H}^{\uparrow}_{\alpha}(X|B)_{\rho} \right) \label{eq:compression}, \quad \textrm{where} \quad \widebar{H}^{\uparrow}_{\alpha}(X|B)_{\rho} = - \min_{\sigma_B \geq 0 \atop \tr[\sigma_B] = 1} \widebar{D}_{\alpha}(\rho_{XB} \| 1_X \otimes \sigma_B) \,.
\end{align}
Our techniques allow us to derive a one-shot lower bound in terms of the measured R\'enyi conditional entropy and asymptotic bounds of the same form as~\eqref{eq:compression} but in terms of the sandwiched R\'enyi conditional entropy. We note that the asymptotic bound that follows from our techniques is worse than the one in~\cite[Section 3.1]{renes23} for general states.

\paragraph*{Acknowledgements:}
MT thanks Hao-Chung Cheng for many insightful and encouraging discussions about the error exponent problem. SB is thankful to Mohammad Hossein Yassaee for sharing some recent works on the error exponent problem for classical channels. SB and TM are supported by the NRF grant NRF2021-QEP2-02-P05 and the
Ministry of Education, Singapore, under the Research Centres of
Excellence program.

\bibliographystyle{ultimate}
\bibliography{ref_auto_mt,ref_manual_mt}

\appendix

\section{Relationships between different bounds on the error exponent} \label{app:hayashi}

We have seen that~\cite{dalai13} and our work establish the following lower and upper bounds for the error exponent of cq channel coding:
\begin{align}
    &\sup_{\alpha \in [\frac12, 1)}  \frac{1-\alpha}{\alpha} \big( \max_{P_X} \min_{\sigma_B} \widetilde{D}_{\alpha}(\rho_{XB} \| \rho_X \otimes \sigma_B) - R \big) \leq E(R) \\
    &\qquad \qquad \leq \sup_{\alpha \in (0, 1)}  \frac{1-\alpha}{\alpha} \big( \max_{P_X} \min_{\sigma_B} \widebar{D}_{\alpha}(\rho_{XB} \| \rho_X \otimes \sigma_B)  - R \big) \label{eq:spherepack}.
\end{align}
Clearly we cannot expect this bound to be tight for general channels since $\widetilde{D}(\rho\|\sigma) \leq \widebar{D}_{\alpha}(\rho\|\sigma)$ with equality only if $\rho$ and $\sigma$ commute as can be seen from~\cite[Lemma 3]{datta13} together with the equality conditions for the Araki-Lieb-Thirring inequality~\cite{hiai94}.
It is a bit less obvious how to compare to the lower bound in~\cite{hayashi07},
\begin{align}
    E(R) &\geq \sup_{\beta \in (0, 1)} (1-\beta) \big( \max_{P_X} \widebar{D}_{\beta}(\rho_{XB} \| \rho_X \otimes \rho_B) - R \big) \\
    &= \sup_{\alpha \in (\frac12, 1)} \frac{1-\alpha}{\alpha} \big( \max_{P_X} \widebar{D}_{2 - \frac{1}{\alpha}}(\rho_{XB} \| \rho_X \otimes \rho_B) - R \big), \label{eq:hayashi2}
\end{align}
since there is no optimisation over the state $\sigma_B$ here.
By introducing a purification of $\rho_{XB}$ in the form $|\rho\rangle_{XX'BC} = \sum_x \sqrt{P(x)} |x\rangle_X \otimes |x\rangle_{X'} \otimes |\rho^x\rangle_{BC}$, we can rewrite~\eqref{eq:spherepack} and~\eqref{eq:hayashi2} using the identities~\cite[Lemma 6]{hayashitomamichel16} 
\begin{align}
    \min_{\sigma_B} \widebar{D}_{\alpha}(\rho_{XB} \| \rho_X \otimes \sigma_B) &= - \widetilde{D}_{\frac{1}{\alpha}}\big(\rho_{XX'C} \big\| \rho_{X}^{-1} \otimes \rho_{X'C} \big) \quad \textnormal{and} \\
    \widebar{D}_{2 - \frac{1}{\alpha}}(\rho_{XB} \| \rho_X \otimes \rho_B) &= - \widebar{D}_{\frac{1}{\alpha}}\big(\rho_{XX'C} \big\| \rho_{X}^{-1} \otimes \rho_{X'C} \big) ,
\end{align}
respectively. From this it becomes evident that the bound in~\cite{hayashi07} is not tight in general since equality only holds if $\rho_{XX'C}$ commutes with $\rho_{X}^{-1} \otimes \rho_{X'C}$. This occurs, for example, when all states $\rho_x$ are pure, since in that case it is easy to verify that $\rho_{XX'}$ commutes with $\rho_{X}^{-1} \otimes \rho_{X'}$ for any input distribution.

\section{Alternative proof of Hayashi's bound}\label{app:Hayashi's-bound-alt}

For two positive semidefinite operators $A, B$ we define 
\begin{align}
D_2^{\iota}(A\| B) := \log Q_2^{\iota}(A\| B),  
\end{align}
with
\begin{align}
Q_2^{\iota}(A\| B):= \tr\left[ A \frac{A}{B}  \right] = \int_0^\infty \tr\left[ A \frac{1}{t+B} A \frac{1}{t+B}  \right] \d t,
\end{align}
if the support of $A$ is included in the support of $B$ and $Q_2^{\iota}(A\|B)=+\infty$ otherwise~\cite{HircheTomamichel2023}.
We note that in the commutative case, $D_2^{\iota}(A\| B)$ reduces to the collision relative entropy. It is shown in~\cite[Example 3.11]{HiaiPetz2013} that $Q_2^{\iota}(A\| B)$ is jointly convex, and then satisfies the data processing inequality. Hence, for any cptp map $\Phi$ we have 
\begin{align}\label{eq:dpi-iota}
D_2^{\iota}(\Phi(A)\| \Phi(B))\leq D_2^{\iota}(A\| B).
\end{align} 
We use this data processing inequality for the proof of Lemma~\ref{lem:int-measurement-error-bound}.
The proof closely follows the footsteps of the proof of~\cite[Lemma 17]{Cheng22}.

\begin{proof}[Proof of Lemma~\ref{lem:int-measurement-error-bound}]
 For any projection $\Pi$ let $\widehat \Pi = \Pi \oplus (1-\Pi)$ be a projection acting on $\cH\oplus \cH$. Define 
\begin{align}
\Phi(X) = \tr\left[X \widehat \Pi\right] \ketbra{0}{0} + \tr\left[ X (1- \widehat \Pi) \right]\ketbra{1}{1}. 
\end{align}
Note that $\Phi$ maps operators acting on $\cH\oplus \cH = \mathbb C^2\otimes \cH$ to \emph{diagonal} operators on $\mathbb C^2$, and is a cptp map. Then, the data processing inequality~\eqref{eq:dpi-iota} yields 
\begin{align}
\tr\left[A\frac{A}{A+B}\right] + \tr\left[B\frac{B}{A+B}\right] & = Q_2^{\iota}(A\|A+B) + Q_2^{\iota}(B\|A+B)\\
& = Q_2^{\iota}\left(A\oplus B\| (A+B)^{\oplus 2}\right)\\
& \geq Q_2^{\iota}\left( \Phi(A\oplus B)\| \Phi\left((A+B)^{\oplus 2}\right)\right)\\
& = \frac{\tr\left[  A\oplus B \widehat \Pi   \right]^2}{  \tr\left[  (A+B)^{\oplus 2} \widehat \Pi  \right] } + \frac{\tr\left[  A\oplus B \left(1- \widehat \Pi\right)   \right]^2}{  \tr\left[  (A+B)^{\oplus 2} \left(1-\widehat \Pi \right) \right] }.
\end{align}
Let $\Pi$ be a projection that satisfies $\|A-B\|_1 = \tr[(A-B)(2\Pi-1)]$. Then, we have 
\begin{align}
\tr[A\oplus B\widehat \Pi] &= \tr[A \Pi] + \tr[B (1-\Pi)] = \frac{1}{2}\left(  \tr[A+B] + \|A-B\|_1  \right),\\
\tr\left[  (A+B)^{\oplus 2} \widehat \Pi  \right]& =  \tr[(A\oplus B) \Pi] + \tr[(A\oplus B) (1-\Pi)] = \tr[A\oplus B] = \tr[A+B].
\end{align}
Therefore,  
\begin{align}
&\tr\left[A\frac{A}{A+B}\right] + \tr \left[B\frac{B}{A+B}\right] \\
&\qquad \geq   \frac{\left( \frac{1}{2}\left(  \tr[A+ B] + \|A-B\|_1  \right)     \right)^2   + \left(   \frac{1}{2}\left(  \tr[A+ B] - \|A-B\|_1  \right)  \right)^2  }{  \tr[A+B] }.
\end{align}
Subtracting $\tr[A+B]$ from both sides and using $\frac{A}{A+B}+\frac{B}{A+B}=1$, we obtain
\begin{align}
& 2\tr\left[A\frac{B}{A+B}\right] \\
&\qquad =\tr\left[A\frac{B}{A+B}\right] + \tr \left[B\frac{A}{A+B}\right] \\
&\qquad \leq   \frac{ \tr[A+B]^2  - \left( \frac{1}{2}\left(  \tr[A+ B] + \|A-B\|_1  \right)     \right)^2   - \left(   \frac{1}{2}\left(  \tr[A+ B] - \|A-B\|_1  \right)  \right)^2  }{  \tr[A+B] }\\
&\qquad =   \frac{ 2\left( \frac{1}{2}\left(  \tr[A+ B] + \|A-B\|_1  \right)     \right)\cdot \left(   \frac{1}{2}\left(  \tr[A+ B] - \|A-B\|_1  \right)  \right)  }{  \tr[A+B] }\\
&\qquad \leq    \tr[A+ B] - \|A-B\|_1, 
\end{align}
where in the last inequality we use $\|A-B\|_1\leq \tr[A+B]$. 
\end{proof}

We now give an alternative proof of~\eqref{eq:hayashi} using the above lemma. Suppose that the codewords $x_1, \dots, x_M$ are chosen independently at random according to an input distribution $P_X$. For the decoder we use~\eqref{eq:def-Pi-Y} with the choice of $Y_x = \rho_x$:
\begin{align}\label{eq:int-pgm-decoder-22}
\Pi_m = \frac{\rho_{x_m}}{\sum_{m'} \rho_{x_{m'}}}.
\end{align}
Then, applying Lemma~\ref{lem:int-measurement-error-bound}, for any $\alpha\in (0,1)$ the average probability of error is bounded by
\begin{align}
\mathbb E_{x_1, \dots, x_M}[p_{\textnormal{err}}] &= \mathbb E_{x_1, \dots, x_M} \tr\left[ \rho_{x_1}   \frac{\sum_{m\neq 1}  \rho_{x_m} }{  \sum_{m'} \rho_{x_{m'}} }  \right] \\
& \leq \frac{1}{2} \mathbb E_{x_1, \dots, x_M} \bigg( \tr[\rho_{x_1}] + \tr\bigg[\sum_{m'\neq 1} \rho_{x_m'}\bigg]  - \bigg\| \rho_{x_1} - \sum_{m'\neq 1} \rho_{x_m'}  \bigg\|_1 \bigg) \label{eq:one-shot-bound-00}\\
& \leq \mathbb E_{x_1, \dots, x_M} \tr\bigg[ \rho_{x_1}^{\alpha} \bigg(  \sum_{m'\neq 1} \rho_{x_m'}  \bigg)^{1-\alpha}  \bigg]\\
& \leq \mathbb E_{x_1} \tr\bigg[ \rho_{x_1}^{\alpha} \bigg( \mathbb E_{x_2, \dots, x_M} \sum_{m'\neq 1} \rho_{x_m'}  \bigg)^{1-\alpha}  \bigg]\\
& = (M-1)^{1-\alpha} \mathbb E_{x_1}  \tr\left[ \rho_{x_1}^\alpha \rho_B^{1-\alpha} \right],
\end{align}
where in the third line we use~\cite[Theorem 1]{audenaert07} and in the fourth line we use the operator concavity of $t\mapsto t^{1-\alpha}$. It is not hard to verify that this bound is equivalent to~\eqref{eq:hayashi}.

We remark that not only the bound of~\eqref{eq:hayashi} on the error exponent can be proven by the decoder~\eqref{eq:int-pgm-decoder-22}, as~\eqref{eq:one-shot-bound-00} shows, the one-shot bounds of~\cite{Cheng22}, particularly~\cite[Theorem 1]{Cheng22}, can also be proven by using this decoder. This means that the decoders defined in terms of the matrix quotient~\eqref{eq:int-matrix-quotient} satisfy some of the main advantages of the standard pretty-good measurement, yet they satisfy an intriguing property stated in Lemma~\ref{lem:quotient}.

\end{document}